\documentclass{amsart}   

\usepackage{booktabs} 
\usepackage{algorithm, algpseudocode, algorithmicx}

\usepackage{float}
\usepackage{amsmath,amsthm,amssymb}
\usepackage{kpfonts}

\newtheorem{theorem}{Theorem}[section]
\newtheorem{corollary}[theorem]{Corollary}
\newtheorem{lemma}[theorem]{Lemma}
\newtheorem{definition}[theorem]{Definition}
\newtheorem{conjecture}[theorem]{Conjecture}

\newcommand{\F}{\mathbb{F}}
\newcommand{\N}{\mathbb{N}}

\newcommand{\poly}{\text{poly}}
\newcommand{\ord}{\text{ord}}
\newcommand{\mb}[1]{\mathbb{#1}}
\newcommand{\mc}[1]{\mathcal{#1}}
\newcommand{\ls}[2]{\left( \frac{#1}{#2} \right)}
\theoremstyle{definition}
\newtheorem{property}[theorem]{Property}
\theoremstyle{remark}

\newtheorem{claim}{Claim}

\begin{document}
\title{Irreducibility and deterministic r-th root finding over finite fields }

\author{Vishwas Bhargava}
\thanks{Visiting Scholar, Centre for Quantum Technologies, NUS, Singapore}
\address{CSE, IIT Kanpur, India}
\email{vishwas1384@gmail.com}

\author{G\'abor Ivanyos}
\address{MTA SZTAKI, Budapest, Hungary}
\email{Gabor.Ivanyos@sztaki.hu}

\author{Rajat Mittal}
\address{IIT Kanpur, India}
\email{rmittal@iitk.ac.in}

\author{Nitin Saxena}
\address{IIT Kanpur, India}
\email{nitin@cse.iitk.ac.in}

\renewcommand{\shortauthors}{Bhargava, Ivanyos, Mittal, Saxena}

\begin{abstract}
Constructing $r$-th nonresidue over a finite field is a fundamental computational problem. A related problem is to construct an irreducible polynomial of degree $r^e$ (where $r$ is a prime) over a given finite field $\F_q$ of characteristic $p$ (equivalently, constructing the bigger field $\F_{q^{r^e}}$). Both these problems have famous randomized algorithms but the derandomization is an open question. We give some new connections between these two problems and their variants.

In 1897, Stickelberger proved that if a polynomial has an odd number of even degree factors, then its discriminant is a quadratic nonresidue in the field. We give an extension of  Stickelberger's Lemma; 
we construct $r$-th nonresidues from a polynomial $f$ for which there is a $d$, such that,  $r|d$  and  $r\nmid\,$\#(irreducible factor of $f(x)$ of degree $d$). Our theorem has the following interesting consequences: (1) we can construct $\F_{q^m}$ in deterministic poly($\deg(f),m\log q$)-time if $m$ is an $r$-power and $f$ is known; (2) we can find $r$-th roots in $\F_{p^m}$ in deterministic poly($m\log p$)-time if $r$ is constant and $r|\gcd(m,p-1)$. 

We also discuss a conjecture significantly weaker than the Generalized Riemann hypothesis to get a deterministic poly-time algorithm for $r$-th root finding.
\end{abstract}

%
%
%
%
%

\maketitle

\section{Introduction}
The problem of finding $r$-th roots in a finite field is to solve $x^r=a$ given an $r$-th residue $a \in \mathbb{F}_q $. Note that, without loss of generality, we can assume $r$ to be prime, otherwise for $r=r_1\cdot r_2$, we can solve the problem iteratively by first solving $x^{r_1}=a$ and then solving $y^{r_2}=x$. Moreover, we can assume $r|(q-1)$, otherwise $x=a^{r^{-1}\mod (q-1)}$ is an easy solution.

It can be shown that $x^r=a$ has a solution iff $a^{\frac{q-1}{r}} =1$ . If $a^{\frac{q-1}{r}}\ne1$ then we call $a$ an {\em $r$-th nonresidue}. Interestingly, the problem of finding an $r$-th nonresidue is {\em equivalent} to that of finding $r$-th roots in $\F_q$ \cite{AMM,Shanks,Tonelli}. This gives a randomized poly-time algorithm for finding $r$-th roots and, thus, solves the problem for practical applications. Also, assuming Generalized Riemann hypothesis (GRH) there is a deterministic poly-time algorithm for finding $r$-th nonresidue in any finite field \cite{Ankeny,Bach,Shoup,Huang}. For a detailed survey see \cite[Chap.7]{BS}.
  
The special case of $r=2$ is particularly well studied. The problem now is to find square-roots in $\mathbb{F}_q$, which is equivalent to finding a {\em quadratic nonresidue} in $\mathbb{F}_q$.  There are other randomized algorithms -- Cipolla's algorithm \cite{Cipolla}, singular elliptic curves based algorithm \cite{Enver}, etc. There are also deterministic solutions for some special cases:
\begin{itemize}
\item Schoof \cite{Schoof} gave an algorithm using point counting on elliptic curves having complex-multiplication to find square-roots of {\em fixed} numbers over prime fields.
\item Tsz-Wo-Sze \cite{Sze} gave an algorithm to take square-roots over $\mathbb{F}_q$, when $q-1 = r^e t$ and $r+t = \poly(\log p)$.
\end{itemize}
However, computing square-roots over finite fields in deterministic polynomial time  is still an open problem. The best known deterministic complexity for this problem is exponential, namely, $\tilde{O}( p^{1/4\sqrt{e}} )$; which is also a bound on the least quadratic nonresidue \cite{Burg}. The distribution of quadratic nonresidues in a finite field is still mostly a mystery; it relates to some interesting properties of the zeta function, see Thm.\ref{thm-npr}.

\smallskip
In 1897, L. Stickelberger \cite{Stick}  proved that if $p$ is a prime, $K$ is an algebraic number field of degree $n$ of discriminant $D$, and integer ring $\mathcal{O}_K$ where the ideal $(p)$ factorizes as $\mathfrak{p}_1\mathfrak{p}_2\mathfrak{p}_3\ldots\mathfrak{p}_s $ into distinct prime ideals then 
\begin{align}\label{SL}
\ls{D}{p} \,=\, (-1)^{n-s} && \mbox{Stickelberger's Lemma} \,.
\end{align}
Equivalently, if the number of even degree irreducible factors of a squarefree $f(x)\mod p$ are odd, then the discriminant of $f$ will be a quadratic nonresidue in $\mathbb{F}_p$. Swan \cite{Swan} and Dalen \cite{kare} gave alternative proofs of Stickelberger lemma. Stickelberger lemma is used in factorization of polynomials over finite fields and in constructing irreducible polynomials of a given degree over finite fields \cite{Gathen01, Swan, Hanson}.

We generalize this idea of constructing quadratic nonresidues from Stickelberger's lemma to constructing $r$-th nonresidues from ``special", possibly reducible, polynomials. Formally, these ``special"  polynomials are over $\mathbb{F}_q$ and satisfy the following factorization pattern,
\begin{property}\label{condition}
Let $r$ be a prime and $f(x) \in \F_q[x]$ be a squarefree polynomial. $f$  satisfies Stickelberger {\em property \ref{condition}} if $\exists d$, such that, $r|d$  and  $r\nmid\,$\#\big(irreducible factor of $f(x)$ of degree $d$\big). 
\end{property}

Our goal is to show that the construction of such a, possibly reducible, polynomial solves many of the open problems. It is somewhat surprising that a reducible polynomial be related so strongly to non-residuosity and irreducibility.

Our first main result relates Property \ref{condition} to the construction of $r$-th nonresidues in {\em any} field above $\F_p$ (equivalently, finding $r$-th roots there).

\begin{theorem} \label{thm-main-1}
Given $\zeta_r \in \mathbb{F}_q$ and any polynomial $f$ satisfying Property \ref{condition}, we can find $r$-th roots in {\em any} finite field of characteristic $p$, in deterministic poly$(\deg(f),\log q)$-time. 
\end{theorem}

We get a stronger result in the case when we have $\F_{p^r}$ available and $r=O(1)$. Even $r=2$ is an interesting special case.

\begin{corollary}\label{cor-main-1}
We can find $r$-th roots in $\F_{p^m}$ in deterministic poly($m\log p$)-time if $r$ is constant and $r|\gcd(m,p-1)$. 
\end{corollary}

Finding an $r$-th nonresidue $a$ in $\F_q$ suffices to construct an extension $\F_{q^r}$. For example, we have $\F_q[a^{1/r}] \cong \F_{q^r}$; equivalently, $X^r-a$ is an irreducible polynomial. However, it is not clear how to find $r$-th nonresidue given $\F_{q^r}$. Anyways, the question of constructing $\F_{q^r}$ efficiently is of great interest \cite{AL86, S90, S94} and still open.

Our second main result relates Property \ref{condition} to the construction of an irreducible polynomial of degree $m$, where $m$ is {\em any} $r$-power.  

\begin{theorem} \label{thm-main-2}
Given a polynomial satisfying Property \ref{condition}, we can construct the field $\F_{q^m}$, for any $r$-power $m$, in deterministic poly$(\deg(f),m\log q)$-time. 
\end{theorem}

Note that, if we are given fields $\F_{q^{m_1}}$ and $\F_{q^{m_2}}$ (for coprime $m_1,m_2$), we can combine them to get the field $\F_{q^{m_1m_2}}$ \cite[Lem.3.4]{S89}. Hence, it is sufficient to be able to construct fields whose sizes are prime powers.

\subsection*{Organization of the paper}
In this paper, the main results and ideas are presented in Sec.\ref{main}. Sec.\ref{prelim} has notation and preliminaries. For concreteness, Sec.\ref{sec-algo} sketches our algorithm for finding an $r$-th nonresidue in any finite field, given a polynomial (in $\F_p[x]$) satisfying Property \ref{condition}. We discuss some special cases of our analysis in Sec.\ref{sec-spl-case}. 

In Sec.\ref{sec-conjs}, we discuss few conjectures; particularly in Sec.\ref{6.2} we introduce a strictly weaker version of Generalized Riemann hypothesis to get poly-time algorithms.


\section{Preliminaries} \label{prelim}

We are going to work in the finite field $\mathbb{F}_q$, where $q = p^d$ for some prime $p$. We will assume that $\mathbb{F}_q$ is specified by a degree 
$d$ irreducible polynomial over $\mathbb{F}_p$. This can be assumed without loss of generality, see \cite[Thm.1.1]{Lens}.  

Given a finite field $\mb{F}_q$ and its extension $\mb{F}_{q^k}$, the multiplicative \textit{norm} of an element $\alpha \in \mb{F}_{q^k}$ is defined as,

\[ N(\alpha) \,=\, N_{\mathbb{F}_{q^k}/\mathbb{F}_{q}} (\alpha)  \,=\, \alpha^{\frac{q^k-1}{q-1}} \,.\]

The following properties of finite fields will be useful (for proofs refer standard texts, eg.~\cite{LN}).
\begin{theorem}[Finite fields]\label{properties}
Given a finite field $\mb{F}_q$ with characteristic $p$ and algebraic closure $\overline{\mb{F}}_p$ ,
\begin{itemize}
\item For any $a \in \overline{\mb{F}}_p$ , $a^{q} = a$ if and only if $a \in \mb{F}_q$.
\item For any $a,b \in \mb{F}_q$, $(a+b)^p = a^p + b^p$.
\item The multiplicative group $\F_q^*$ is cyclic.  
\item Any polynomial $f \in \mb{F}_q[x]$ of degree $k$ has at most $k$ roots in $\mb{F}_q$.
The notation $\mathcal{Z}(f)$ will be used to denote the set of zeros of polynomial $f(x)$.
\end{itemize} 
\end{theorem}

We are interested in finding $r$-th nonresidue in $\mb{F}_q$ for a prime $r$. An element $a \in \mb{F}_q$ is called an $r$-th nonresidue iff $x^r = a$ has no roots in $\mb{F}_q$. This possibility is there only if $r|(q-1)$.
In that case, $a$ is an $r$-th nonresidue iff $a^{\frac{q-1}{r}} \neq 1$ \cite{BS}. Using this characterization, the following lemma constructs an $r$-th nonresidue in $\mb{F}_q$ 
given an $r$-th nonresidue in $\mb{F}_{q^k}$. 

\begin{lemma}[Projection]\label{projection}
Let $r$ be a prime which divides $q-1$. Then, 
$\alpha \in \mb{F}_{q^k}$ is an $r$-th nonresidue iff $N_{\mathbb{F}_{q^k}/\mathbb{F}_{q}}(\alpha)$ is an $r$-th nonresidue in $\mb{F}_q$. 
\end{lemma}

\begin{proof}
We know that, 
\[ N_{\mathbb{F}_{q^k}/\mathbb{F}_{q}} (\alpha) = \prod_{i=1}^{k-1} \alpha^{q^i} = \alpha^{\frac{q^k-1}{q-1}}. \]
  
Also, $\alpha \in \mathbb{F}_{q^k}$ is a $r$-th nonresidue iff $\alpha^{\frac{q^k-1}{r}} \neq 1$.

Hence, the proof follows from the bi-implication,
\[ \alpha^{\frac{q^k-1}{r}} \neq 1 \Longleftrightarrow \Big({\alpha}^{\frac{q^k-1}{q-1}}\Big)^{\frac{q-1}{r}} = \Big( N_{\mb{F}_{q^k}/\mb{F}_{q}} (\alpha) \Big)^{\frac{q-1}{r}} \neq 1. \]
\end{proof}

We can define a multiplicative character-- $\chi_{r} (a) := a^{\frac{q-1}{r}}$ --of $\F_q^*$. Notice that $\chi_{r} (a) \neq 1$ iff $a$ is an $r$-th nonresidue in $\mb{F}_q$. 
Multiplicativity follows from the definition, i.e.,
\[ \chi_{r}(ab) = \chi_{r}(a) \chi_{r}(b) \,.\]

Since $a^{q-1} = 1$, $\chi_r(a)$ is an $r$-th root of unity. We will denote a \textit{primitive} $r$-th root of unity by $\zeta_r$.  


Since $\F_q^*$ is cyclic and $r\mid q-1$, we have that $\zeta_r$ exists in $\mb{F}_q$. Note that $\zeta_r^i$, $i\in\F_r^*$, are the $(r-1)$ primitive $r$-th roots of unity in $\F_q$.

One of the central algebraic tool used in our analysis is the \textit{resultant} of two polynomials. Let $f(x) = a_m x^m + a_{m-1} x^{m-1} + \cdots + a_0$ and $g(x) = b_n x^n + b_{n-1} x^{n-1} + \cdots + b_0$ 
be two polynomials over a field $\mb{F}$. 


\begin{definition}[Resultant]\label{resultant}
One way to define resultant of the two polynomials $f,g \in \mb{F}[x]$ is by invoking the zeros of the polynomials (in $\overline{\F}$),

\[ R(f,g) \,:=\, a_m^m b_n^n \mathop{\prod_{\alpha\in \mc{Z}(f)}}_{\beta \in \mc{Z}(g)} (\alpha -\beta) \,=\, a_m^m \prod_{\alpha \in \mc{Z}(f)} g(\alpha) \,. \]

\end{definition}

We will use the following properties of resultant (for proof see \cite[Chap.1]{LN}). In fact, the property (3) in Lem.\ref{lem-res} can be taken as the general definition of resultant, as it makes the resultant efficient to compute even when the base ring is not a field.

\begin{lemma}[Properties of $R(\cdot)$]\label{lem-res}
Given  polynomials $f, g, h \in \mathbb{F}[x]$, we have that,
\begin{enumerate}
\item $R(f,g)\in \mathbb{F}$. 
\item Resultant is multiplicative, $R(fh,g) = R(f,g)\cdot R(h,g)$.
\item Resultant is the determinant of {\em Sylvester matrix} of order $m+n$ and, thus, can be computed in time $\tilde{O}((m+n)^{\omega_0})$, where ${\omega_0}\le 2.373$ is the exponent of matrix multiplication.
\begin{frame}
\footnotesize
\arraycolsep=1pt 
\medmuskip = 0.5mu 
\[
R(f,g) \,=\, \left|
\begin{array}{cccccccc}
a_m & 0 & \cdots & 0 & b_n & 0 & \cdots & 0 \\
a_{m-1} & a_m & \cdots & 0 & b_{n-1} & b_{n} & \cdots & 0 \\
a_{m-2} & a_{m-1} & \ddots & 0 & b_{n-2} & b_{n-1} & \ddots & 0 \\
\vdots & \vdots & \ddots & a_m & \vdots & \vdots & \ddots & b_n \\
\vdots & \vdots & \cdots & a_{m-1} & \vdots & \vdots & \cdots & b_{n-1} \\
a_0 & a_1 & \cdots & \vdots & b_0 & b_1 & \cdots & \vdots \\
0 & a_0 & \ddots & \vdots & 0 & b_0 & \ddots & \vdots \\
\vdots & \vdots & \ddots & a_1 & \vdots & \vdots & \ddots & b_1 \\
0 & 0 & \cdots & a_0 & 0 & 0 & \cdots & b_0 \\
\end{array}
\right| \,.
\]
\end{frame}
\end{enumerate}
\end{lemma}

Another tool, closely related to resultant, is called the \textit{discriminant}. 

\begin{definition}[Discriminant]\label{discriminant}
The discriminant of a polynomial $f \in \mb{F}[x]$ with roots $\mc{Z}(f) = \{\alpha_1, \alpha_2, \cdots, \alpha_m \}$ is defined by, 

\[ \Delta (p) \,:=\, a_m^{2m-2} \prod_{1\leq i <j \leq m} (\alpha_i -\alpha_j)^2 \,. \]
\end{definition}

It is known that $\Delta (f) = (-1)^{m(m-1)/2} a_m^{-1}\cdot R(f, f')$ \cite[Eqn.1.11]{LN}, where $f'$ is the formal derivative of $f$. Hence, $\Delta(f) \in \mb{F}$ and it can be computed in poly($m$) field operations. 

Note that although resultant (resp.~discriminant) is defined in terms of the zeros of the polynomials, it can be computed without the knowledge of the zeros. This relationship between the zeros and the coefficients is very useful computationally. 


\section{Main results} \label{main}


We will prove the main theorems in this section. We are interested in finding $r$-th nonresidue in the finite field $\mb{F}_q$. 
So we will assume that $r \mid q-1$ in Sec.\ref{cor-main-1} and Sec.\ref{sec-red-poly}. 
Moreover, for $r=2$ we can assume that $4|(q-1)$, otherwise $-1$ is a quadratic nonresidue and we are done.
 
Our first step will be to construct an $r$-th nonresidue using an irreducible polynomial $f$ of degree divisible by $r$.

\subsection{From an irreducible polynomial $f$ -- Proof of Cor.\ref{cor-main-1}}

Given an irreducible polynomial $f(x) \in \mb{F}_q[x]$ of degree $d = rk$, define the following polynomial (inspired from {\em Lagrange resolvents}):
 
\[ L_{f,r} := \sum_{i=0}^{r-1} x^{(q^k)^i} {\zeta_r}^i \mod f . \] 

The following theorem finds an $r$-th nonresidue in $\mb{F}_q$ using $f$.

\begin{theorem}[Irreducibility to nonresiduosity]\label{key}
Let $f(x) \in \mb{F}_q[x]$ be an irreducible polynomial of degree $d=rk$ and $\gcd(2,r)\cdot r \mid q-1$. If $L_{f,r} := \sum_{i=0}^{r-1} x^{q^{k \cdot i}} {\zeta_r}^i \mod f$, then

\[ \big( L_{f,r} \big)^{\frac{q^d-1}{r}} \,=\, \zeta_r^{-1} \,.\]

This implies that $L_{f,r}$ is an $r$-th nonresidue in $\mb{F}_{q^d}= \mb{F}_{q}[x]/\langle f \rangle$. Also, $N_{\mathbb{F}_{q^d}/\mathbb{F}_{q}}(L_{f,r})$ is an $r$-th nonresidue in $\mb{F}_q$.  
\end{theorem}

\begin{proof}

We know that $L_{f,r} \in \mathbb{F}_{q^d}$ and $\zeta_r \in \mb{F}_q$. Taking the $q^k$-th power,

\begin{align*} 
{(L_{f,r})}^{q^k} &= \Big(\sum_{i=0}^{r-1} {x}^{q^{ki}} {\zeta_r}^i \Big)^{q^k}\\
&= \sum_{i=0}^{r-1} {x}^{q^{k \cdot (i+1)}} {\zeta_r}^i = {\zeta_r}^{-1} \cdot L_{f,r} \,.
\end{align*}

Using the above equation,  
\begin{align*}
{(L_{f,r})}^{\frac{q^d-1}{r}} &={{(L_{f,r})}^{(\frac{q^d-1}{q^k-1})\cdot {(\frac{q^k-1}{r})}}} \\
~&= {{(L_{f,r})}^{(1+q^k + q^{2k} \ldots + q^{(r-1)k})\cdot {(\frac{q^k-1}{r})}}} \\
~&= {{(L_{f,r}\cdot \zeta_r^{-1} L_{f,r}\cdot \zeta_r^{-2} L_{f,r}\cdots \zeta_r^{-(r-1)} L_{f,r}    )}^{(\frac{q^k-1}{r})}} \\
~&= \Big((L_{f,r})^r (\zeta_r^{-1})^{{r}\choose{2}}\Big)^{(\frac{q^k-1}{r})} \,.
\end{align*}

When $r$ is an odd prime, $(\zeta_r^{-1})^{{r}\choose{2}}$ is $1$. If $r$ is $2$ then we have $4|(q-1)$, thus the factor of $-1$ can be ignored.
Simplifying,

\begin{align*}
{(L_{f,r})}^{\frac{q^d-1}{r}} &=\Big((L_{f,r})^r \Big)^{(\frac{q^k-1}{r})} \\
~&=(L_{f,r})^{({q^k-1})}=\zeta_r^{-1} \,.
\end{align*}

By definition of $r$-th nonresidue, this implies that $L_{f,r}$ is an $r$-th nonresidue in $\mb{F}_{q^d}$. Applying Lem.~\ref{projection}, 
we get that $N_{\mathbb{F}_{q^d}/\mathbb{F}_{q}}(L_{f,r})$ is an $r$-th nonresidue in $\mb{F}_q$. 
\end{proof}

Thm.~\ref{key} gives the Cor.\ref{cor-main-1}.
\begin{proof}[Proof of Cor.\ref{cor-main-1}]
Since $\mb{F}_{p^m}$ is specified by an irreducible polynomial of degree $m$ (and we know $r|\gcd(m,p-1)$), we get an $r$-th nonresidue by Thm.\ref{key} if we can find $\zeta_r$ in $\F_p$. The latter can be done using Pila's algorithm based on arithmetic algebraic-geometry \cite[Thm.D]{Pila}. Once we have an $r$-th nonresidue one gets an $r$-th root finding algorithm \cite{Shanks, Tonelli}.
\end{proof}

Thm.~\ref{key} also gives us a way to construct $r$-th nonresidue, in $\mathbb{F}_{p^n}$ for any $n$, using an irreducible polynomial of degree divisible by $r$.

\begin{corollary}[Any field]\label{charp}
Suppose we have an irreducible $f \in \mb{F}_q [x]$ with degree $d=rk$ and $\zeta_r\in\F_q$, where $\mb{F}_q$ has characteristic $p$. 
Then, we can find $r$-th nonresidue in any finite field $\F_{q'}$ of characteristic $p$ (assuming $r|(q'-1)$).  
\end{corollary}
\begin{proof}
Let $\F_{p^m}$ be the {\em smallest} subfield of $\F_q$, with $r|(p^m-1)$. Using Thm.~\ref{key} \& Lem.~\ref{projection} on $f$, we can find an $r$-th nonresidue in $\mb{F}_{p^m}$.

Now consider the given field $\F_{q'}$ with, say, $p^{m\ell}$ elements (since $r|q'-1$, $\ell\in\N$). It has a subfield $\mb{F}'$ of size $p^m$, and so by \cite[Thm.1.2]{Lens}, 
we also get an $r$-th nonresidue in $\mb{F}'$, say $a$. We intend to lift this nonresidue to the bigger field $\F_{q'}$; to do that we consider two cases.
\begin{itemize}
\item Case 1: If $r\nmid \ell$ then $a$ is an $r$-th nonresidue in $\mb{F}_{q'}$. Because,

\[ a^{\frac{q'-1}{r}} = (a^{\frac{p^m-1}{r}})^{\frac{q'-1}{p^m-1}} = (\zeta_r^{-1})^{\frac{q'-1}{p^m-1}} \neq 1 \,.  \]

Last inequality holds because $\frac{q'-1}{p^m-1} = \frac{p^{m\ell}-1}{p^m-1}$ is not divisible by $r$.
 
\item Case 2: If $r\mid \ell$ then we have an irreducible polynomial that defines $\F_{p^{m\ell}}$ and on that we can apply Thm.\ref{key} to get an $r$-th nonresidue in $\mb{F}_{q'}$.
\end{itemize} 
\end{proof}

The following lemma relates $N_{\mathbb{F}_{q^d}/\mathbb{F}_{q}}(g)$ to the resultant $R(f,g)$ when $f$ is irreducible.  

\begin{lemma}[Resultant as a norm]\label{norm}
If $f$ is an irreducible polynomial of degree $d$ in $\mb{F}_q[x]$, then
\[    R(f,g) = N_{\mathbb{F}_{q^d}/\mathbb{F}_{q}}(g) . \]
\end{lemma}

\begin{proof}
We know that the roots of polynomial $f$ are $\mc{Z}(f) = \{ \alpha, \alpha^q, \cdots, \alpha^{q^{d-1}}\}$.
Using the definition of resultant,
\begin{align*}
R(f,g) &= \prod_{\alpha \in \mc{Z}(f)} g(\alpha)  \\
~ &=  \prod_{i=0}^{d-1} g(\alpha^{q^i}) \\
~ &=  \prod_{i=0}^{d-1} g(\alpha)^{q^i} \\
~ &=  g(\alpha)^{\sum_{i=0}^{d-1} q^i} \\
~ &=  N_{\mathbb{F}_{q^d}/\mathbb{F}_{q}}(g) \,.
\end{align*}
\end{proof}

Using Thm.~\ref{key} and Lem.~\ref{norm}, we immediately get the following information about the quadratic character of the resultant of the Lagrange resolvent,

\begin{corollary}[Resultant of resolvent]\label{chi-r}
In the notation of Thm.\ref{key}, $R(L_{f,r},f)$ is an $r$-th nonresidue in $\mb{F}_q$. 

In particular, $\chi_{r} (R(L_{f,r},f)) = \zeta_r^{-1}$.     
\end{corollary}

\subsection{From a reducible polynomial $f$ -- Proof of Thm.\ref{thm-main-1}}\label{sec-red-poly}

We will look at the case of reducible polynomials now. The Thm.\ref{thm-main-1} shows that a reducible polynomial satisfying Property \ref{condition} will give us an $r$-th nonresidue. 
Note that an irreducible polynomial of a degree divisible by $r$ will trivially satisfy Property \ref{condition}.


\begin{proof}[Proof of Thm.\ref{thm-main-1}]
By distinct degree factorization \cite[Thm.7.5.3]{BS}, the polynomial $f$ can be decomposed as $f = h_1h_2 \cdots h_n$, s.t.,
\begin{itemize}
\item For all $i$, $h_i$ has irreducible factors of same degree.
\item For all $i\ne j$, irreducible factors of $h_i$ and $h_j$ have different degree.
\end{itemize}    

We know that $f$ satisfies Property~\ref{condition}. So, the distinct degree factorization guarantees a factor $h_i = f_1 f_2 \cdots f_{r'}$ of $f$ such that, 
\begin{itemize}
\item $f_i$'s are irreducible of degree $d = rk$.
\item $r \nmid r'$.
\end{itemize}
For convenience we shall denote $f_1 f_2 \cdots f_{r'}$ as $f$ from now on.
Define $g(x)$ to be the Lagrange resolvent inspired polynomial, 
\[ g(x) \,:=\, \sum_{i=0}^{r-1} x^{q^{ki}} {\zeta_r}^i  \,.\]

We will show that $R(f,g \bmod f)$ is an $r$-th nonresidue in $\mb{F}_q$. Here $g\bmod f$ refers to some representative in $\mb{F}_{q^d}[x]$. We will now show that the resultant is independent of the representative chosen.

\begin{claim}\label{res-mod}
Let $f,g$ be two polynomials over any field. Then,
\[ R(f,g \bmod f) = R(f,g) \,.\]
\end{claim}
\begin{proof}
Let $g' := g \bmod f$ be a representative. Using the definition of resultant,
\begin{align*}
R(f,g') &= \prod_{\alpha \in \mc{Z}(f)} g'(\alpha) \\
~ &= \prod_{\alpha \in \mc{Z}(f)} g(\alpha) \quad [\because g'(\alpha)=g(\alpha)] \\
~ &= R(f,g) \,.
\end{align*}
\end{proof}

Clm.\ref{res-mod} implies that,
\[ R(f, g\bmod f) = R(f,g) = \prod_{i=1}^{r'} R(f_i,g) =  \prod_{i=1}^{r'} R(f_i,g \bmod f_i) . \] 

Since $\chi_r$ is multiplicative, we have,

\[ \chi_r(R(f, g \bmod f)) = \prod_{i=1}^{r'}  \chi_r(R(f_i,g \bmod f_i)) = (\zeta_r^{-1})^{r'} . \] 

The last step follows from Cor.~\ref{chi-r} and the fact that $f_i$ are irreducible. Since $r \nmid r'$, we get $\chi_r(R(f,g \bmod f)) \neq 1$ and hence $R(f,g \bmod f)$ is 
an $r$-th nonresidue in $\mb{F}_q$.

The last statement of the theorem (about fields of characteristic $p$) follows in the same way as in the proof of Cor.~\ref{charp}. 

The time complexity is straightforward and further discussed in Sec.\ref{sec-algo}.

\end{proof}

\subsection{Constructing fields -- Proof of Thm.\ref{thm-main-2}} \label{sec-teich}

The result in the previous subsection required the existence and knowledge of $\zeta_r$. 
Now we would like to eliminate those assumptions, hence we will remove the assumption $r|q-1$. First, we will show that if we have a reducible polynomial $f$ satisfying Property \ref{condition} then we can construct $\F_{q^r}$ (equivalently, we can construct an irreducible polynomial of degree $r$). The concepts that we will use are inspired from the proof of \cite[Thm.5.2]{Lens}.

The starting idea is to work with a ``virtual'' $\zeta_r$, i.e.~define the ring $\F_q[\zeta] := \F_q[Y]/\langle \varphi_r(Y) \rangle$, where $\varphi_r(Y) := \sum_{0\le i\le r-1} Y^i$, and let $\zeta$ be the residue-class of $Y \bmod \varphi_r(Y)$ in that ring. 
Let $e$ be the smallest positive integer such that $r|q^e-1$, in other
words, the multiplicative order of $q$ modulo $r$. 
Then
$\varphi_r(Y)$ completely splits over $\F_{q^e}$ as 
$$\varphi_r(Y) \,=\, \prod_{i\in \F_r^*} (Y-\eta^i) \,,$$ 
where $\eta\in\F_{q^e}$ is a primitive $r$-th root of unity, 
but we may not have access to $\eta$ and in general not even to $\F_{q^e}$. 
So we will do computations over the {\em ring} 
$\F_q[\zeta]$ and try to construct the field $\F_{q^r}$. 

Clearly, $\zeta$ has order $r$ in the unit group $\F_q[\zeta]^*$. For each integer $a\in\F_r^*$ there is a unique ring automorphism $\rho_a$ of $\F_q[\zeta]$ that fixes $\F_q$ and maps $\zeta\mapsto \zeta^a$. The set $\{\rho_a \,\mid\, a\in\F_r^*\} =: \Delta$ forms a group (under map composition) that is isomorphic to $\F_r^*$. If we consider the elements of the ring fixed under $\Delta$ then we get back $\F_q$, i.e.~$\F_q[\zeta]^\Delta = \F_q$ \cite[Prop.4.1]{Lens}. 

Like Sec.\ref{sec-red-poly}, suppose we have an $f = f_1 f_2 \cdots f_{r'} \in\F_q[x]$ with $f_i$'s being irreducibles of degree $d = rk$ and $r \nmid r'$. When we move to $\F_{q^e}$, $f_i$ factors into $\ell:=\gcd(k,e) = \gcd(d,e)$ many irreducibles each of degree $d/\ell = kr/\gcd(k,e) =: k'r$. Since $r\nmid e$, we have that $r\nmid \ell$.

Our ring $\F_{q}[\zeta]$ is a semisimple algebra that decomposes as:

$$\F_{q}[\zeta] \,\cong\, \varprod_{i \,\in\, \F_r^*/\langle q\rangle}\F_{q^e}[Y]/\langle Y-\eta^i\rangle \,,$$ 

and the proof given in Sec.\ref{sec-red-poly} holds simultaneously over each of 
the component fields ($\cong\F_{q^e}$) of $\F_{q}[\zeta]$. Hence, simply by Chinese remaindering, we get the equality:
\begin{equation}\label{eqn-over-zeta}
R(f, g \bmod f)^{\frac{q^e-1}{r}} \,=\, \zeta^{-r'\ell} \,,
\end{equation}
where, as expected, $g(x)$ is the following Lagrange resolvent over 
$\F_{q}[\zeta]$,
\[ g(x) \,:=\, \sum_{i=0}^{r-1} x^{q^{ek'i}} {\zeta}^i  \,.\]
(Also, note that we are now computing mod and resultant over the base ring $\F_{q}[\zeta]$.)

\smallskip\noindent {\bf Teichm\"uller subgroup.}
Let $r''$ be an integer representative for $(r'\ell)^{-1}\bmod r$. Let $q^e-1=ur^t$ such that $r\nmid u$ and $t\ge1$. Define $\delta := R(f, g \bmod f)^{ur''}$. Then, by Eqn.\ref{eqn-over-zeta}, we have $\delta^{r^{t-1}} = \zeta^{-1}$. In particular, $\delta$ has order $r^t$ in $\F_q[\zeta]^*$. Define a function $\omega$ that maps any integer $a$ to $a^{r^{t-1}} \bmod r^t$. Note that, by binomial expansion, $(a+r)^{r^{t-1}} \equiv a^{r^{t-1}} \bmod r^t$. In other words, value of $\omega(a)$ only depends on $a\bmod r$. Now we come to the key definition, inspired from \cite{Lens},
$$c \,:=\, \left(\prod_{a\in[r-1]} \rho_a^{-1}\left(
\delta^{\omega(a)}\right)\right)  \,.$$

The following properties can be easily verified:
\begin{itemize}
\item $c^{r^{t-1}} = \zeta$, 
\item $c$ has order $r^t$ in $\F_q[\zeta]^*$, and
\item for all $\rho_b\in\Delta$, $\,\rho_b(c) = c^{\omega(b)}$.
\end{itemize}

At this point recall the definition of {\em Teichm\"uller subgroup} w.r.t. $\F_q$:
\begin{align*}
T_{\F_q} \,:=\, \big\lbrace \epsilon\in\F_q[\zeta]^* \,\mid\, & \epsilon\text{ has $r$-power order, and } \\
& \forall\rho_a\in\Delta,\, \rho_a(\epsilon) = \epsilon^{\omega(a)} \big\rbrace \,.
\end{align*}

By the properties above and invoking \cite[Thm.5.1]{Lens}, we can deduce that $c$ is a generator of $T_{\F_q}$.

Consider the extension ring $\F_q[\zeta][c^{1/r}] := \F_q[\zeta][X]/\langle X^r-c \rangle$, where $c^{1/r}$ is the residue class of  $X\bmod X^r-c$ in the ring. By \cite[Prop.4.3]{Lens} we have: $\forall b\in\F_r^*$, $\rho_b$ extends uniquely to a ring automorphism of $\F_q[\zeta][c^{1/r}]$ such that $c^{1/r}\mapsto (c^{1/r})^{\omega(b)}$. Thus, $\Delta$ can now be seen as a {\em group of ring automorphisms} of $\F_q[\zeta][c^{1/r}]$.

Now we have the following nice way to construct a field extension.

\begin{theorem}[Field extension]\label{thm-Fq-to-r}
The fixed subring $\F_q[\zeta][c^{1/r}]^\Delta$ is isomorphic to $\F_{q^r}$. Moreover, given $f$, $\F_{q^r}$ can be constructed in deterministic poly($\deg(f),r\log q$)-time.
\end{theorem}
\begin{proof}
It directly follows from \cite[Thm.5.1]{Lens} that $\F_q[\zeta][c^{1/r}]^\Delta \cong \F_{q^r}$.

From the above discussion it can be seen that, given $f$, we can compute $c$. Hence, we have a representation of the ring $\F_q[\zeta][c^{1/r}]$ in terms of a linear basis $\mathcal{B}$ over $\F_q$ (\& their multiplication relations). Because of the properties of $\rho_b(\zeta)$ and $\rho_b(c^{1/r})$ we also have a description of the action of $\Delta$ on $\F_q[\zeta][c^{1/r}]$ in terms of $\mathcal{B}$. Thus, we can compute the fixed subring $\F_q[\zeta][c^{1/r}]^\Delta$ efficiently. It is straightforward to get the time complexity estimate.
\end{proof}

\begin{proof}[Proof of Thm.\ref{thm-main-2}]
From $f$, by Thm.\ref{thm-Fq-to-r}, we can get an irreducible polynomial $g$ over $\F_q$ of degree $r$. Let $m$ be any $r$-power. Then, by \cite[Thm.1.1]{Lens}, we can construct $\F_{q^m}$ using $g$ efficiently.
\end{proof}


\section{Algorithm} \label{sec-algo}
For concreteness, we state our algorithm (Algo.\ref{alg1}) for constructing  $r$-th nonresidue in this section.  The proof of correctness for this algorithm follows directly from Thm.\ref{thm-main-1}.

The input to this algorithm is a polynomial $f(x) \in \mb{F}_p[x]$ satisfying Property  \ref{condition}, $\zeta_r \in \mathbb{F}_p$, and the finite field $\mathbb{F}_{q'} $ of characteristic $p$ where we want to construct $r$-th nonresidue. The algorithm outputs an $r$-th nonresidue in $\mathbb{F}_{q'} $.

Note that, since $f(x)$ satisfies Property \ref{condition}, wlog (by the distinct degree factorization) $f = f_1 f_2 \cdots f_{r'}$ such that, 
\begin{itemize}
\item $f_i$'s are irreducible of degree $d = rk$, and
\item $r \nmid r'$.
\end{itemize} 
\begin{algorithm}
\caption{BIMS}
\label{alg1}
\begin{algorithmic}[1]
\Statex \hspace{-4mm} \textbf{Input} : $f(x),\, \zeta_r \in \mathbb{F}_p,\, \F_{q'} $,  where $q'=p^n$. 
\Statex \hspace{-4mm} \textbf{Output} : $r$-th nonresidue in $\F_{q'}$. 
\If { ($r|n$) } 
	\State Define $g(x)= \sum_{i=0}^{r-1} x^{v^i} {\zeta_r}^i  \mod h(x)$ 				\Comment{where $v:= p^{n/r}$ and $h(x)$ is the minimal polynomial of $\mathbb{F}_{q'}$ over $\mathbb{F}_p$ .}
	\State Output $g(x)$.
\Else
	\State Define $g(x)= \sum_{i=0}^{r-1} x^{v^i} {\zeta_r}^i  \mod f(x)$ 				\Comment{where $v:= p^{k}$.} 
	\State Output  $R(g(x),f(x))$.
\EndIf
\end{algorithmic}
\end{algorithm}

\noindent {\em Time complexity analysis-} 

One can refer to \cite{S05} for basic arithmetic operations.
Polynomial computation in Step 2, takes time $\tilde{O}(rn\log{p}\log q')$  using repeated squaring. Similarly, Step 5 can be done in  $\tilde{O}(r^2k\log{p}$ $\deg(f))$. The most expensive part of the algorithm is the resultant computation in Step 6. The same can be done in time $\tilde{O}(\deg(f)^{\omega_0}\log{p})$, where $\omega_0 <2.373$.

\section{Some special case applications} \label{sec-spl-case}
\subsection{The special case of $r=2$}

Notice that for $r=2$, we have $\zeta_2=-1$ available in any finite field with odd characteristic. Thus, using Thm.\ref{thm-main-1} and an $f$ (Property \ref{condition}), we can construct a quadratic nonresidue. The same can also be calculated using Stickelberger lemma directly.

A striking difference, in the case of $r=2$, is that using Stickelberger lemma (Eqn.\ref{SL}) discriminant is the quadratic nonresidue. This implies that over even degree finite field extensions, the derivative of the minimal polynomial of the extension is a quadratic nonresidue. We formally state this property below.
 
\begin{lemma}[Derivative]
Given a finite field $\F_{q^d} = \mb{F}_{q}[x]/\langle f \rangle$ with even $d=\deg(f)$ and $4|(q-1)$, $f'$ is a quadratic nonresidue in $\mb{F}_{q}[x]/\langle f \rangle$. 
\end{lemma}
\begin{proof}
Using Stickelberger lemma (Eqn.\ref{SL}) we know that the discriminant is a quadratic nonresidue in $\mathbb{F}_q$. Since,
$$\Delta (f) \,=\, (-1)^{d(d-1)/2} a_d^{-1}\cdot R(f,f') \,,$$  
where $a_d=1$ is the leading coefficient of $f(x)$, we can deduce that $R(f,f')$ is a quadratic nonresidue in $\mathbb{F}_q$.

Using  Lem.\ref{norm} we get that $N_{\F_{q^d}/\F_q}(f')$ is a quadratic nonresidue in $\mathbb{F}_q$, and using Lem.\ref{projection} we get that $f'(x)$ is a quadratic nonresidue in $\mb{F}_{q}[x]/\langle f \rangle$.
\end{proof}

\subsection{Cases for which $\zeta_r$ is known}
Since our first main theorem, Thm.\ref{thm-main-1}, requires $\zeta_r$, in this section we state some known methods to construct the same.

One of the most significant results on this is by Pila \cite{Pila}. He generalized Schoof's \cite{Schoof} elliptic curve point-counting algorithm to Fermat curves, and as an application gave an algorithm for factoring the $r$-th cyclotomic polynomial over $\mathbb{F}_p$. The algorithm is deterministic and runs in time polynomial in $\log{ p}$ for a fixed $r$. If $r|p-1$ then the factorization of the $r$-th cyclotomic will give us $\zeta_r \in \mathbb{F}_p$.

A limitation of Pila's algorithm is that it can give us $\zeta_r$ only in prime fields. Below we state few results that can give $\zeta_r$ in extensions of prime fields.
 
The following theorem by Bach, von zur Gathen and Lenstra \cite{BGL} gives an elegant way to construct  $\zeta_r \in \mathbb{F}_q $ using ``special'' irreducible polynomials.   

\begin{theorem}\cite[Thm.2]{BGL} \label{BGL}
Given two prime numbers $p$ and $r$, the $h=\ord_r(p)$, the explicit data for $\mathbb{F}_{p^h}$; and given for each prime $\ell|(r-1)$ but not dividing $h$, an irreducible polynomial $g_\ell$ of degree $\ell$ in $\mathbb{F}_p[X]$, there is a deterministic poly($rh\log(p)$)-time algorithm to construct a primitive $r$-th root of unity in $\mathbb{F}_{p^h}$.
\end{theorem}

We immediately get the following.

\begin{corollary}[Inspired by BGL \cite{BGL}]
Let prime $r|q-1$. If for each prime $\ell \mid (r-1)$ we are given an irreducible polynomial $h_\ell\in\mathbb{F}_q[x]$ of degree divisible by $\ell$, then we can construct $\zeta_r\in\F_q$.   
\end{corollary}
\begin{proof}
Using $h_\ell$ we can construct an irreducible polynomial of degree $\ell$ \cite[Thm.1.1]{Lens}. Using Thm.\ref{BGL} on these, we can construct $\zeta_r \in \mathbb{F}_q$. 
\end{proof}

There are also some other methods for finding $\zeta_r \in \mathbb{F}_q$ that are based on the factorization pattern of $q-1$. We present one such result and its proof. 

\begin{theorem}[Tsz-Wo Sze \cite{Sze}]
We can find $\zeta_r \in \mathbb{F}_q$ if $q-1 = r^e t$, where $r+t = \poly(\log q)$. 
\end{theorem}  
\begin{proof}
The number of elements whose order is not a multiple of $r$ is $t$. So if we take $t+1$ elements in $\mathbb{F}_q$, this will give us an element $a$ that has order a multiple of $r$. Then, $a^{t}$ is an element with an $r$-power order. Let $\ord(a^t)=: r^s$, where $s \geq 1$. Finally, $a^{t\cdot r^{s-1}}$ is an element of order $r$ in $\mathbb{F}_q$.
\end{proof}

\subsection{Necessary condition for the irreducibility of a polynomial}
Our analysis provides a necessary condition for checking  irreducibility of a polynomial.

\begin{lemma}\label{irrd}
If $f\in \mathbb{F}_q[x]$ is irreducible and prime $r | \deg(f)$ with $\gcd(2,r)\cdot r\mid (q-1)$, then $ R(L_{f,r},f(x))$ is an $r$-th nonresidue in $\mathbb{F}_q$. 
\end{lemma}
\begin{proof}
This follows directly from Thm.\ref{key}.
\end{proof}


Lem.\ref{irrd} for $r=2$ is used by von zur Gathen in his paper to prove properties about irreducible trinomials \cite[Cor.3]{Gathen01}. We hope that this generalized lemma gives conditions that can help construct additional polynomial families.

\section{Some conjectures} \label{sec-conjs}

\subsection{Finding polynomials satisfying Property \ref{condition}} \label{6.1}

A natural question that arises from our analysis is: How can one construct a polynomial satisfying Property \ref{condition}?
An approach can be to come up with a polynomial family  $\mathcal{F}$ such that at least one of the polynomial in $\mathcal{F}$ satisfies Property \ref{condition}. We leave the construction of such a polynomial family as an open question.

This question for $r=2$ will also be very interesting. For $r=2$, if we can construct a polynomial satisfying Property \ref{condition} then its discriminant will be a quadratic nonresidue by Stickelberger's lemma.   

A well studied polynomial family for such properties are trinomials. {\em Trinomials} are univariate polynomials with sparsity three:
$$\mathcal{T}_{(n,k,a,b)} \, =\, \{x^n+ax^k+b \,\mid\, n>k>0;  a,b \in \mathbb{Z}^* \} \,.$$
 An elegant property of trinomials is the closed form expression for their discriminant and, thus, it can be computed efficiently. (Even if the degree of the trinomial is exponential.)
 
\begin{theorem}[Swan \cite{Swan}] \label{trinomial}
Let $n>k>0$. Let $d=\gcd(n,k)$ and $n=n_1d, k=k_1d$ . Then,
\begin{align*}
\Delta(x^n+ax^k+b) \,=\, (-1)^{n(n-1)/2}b^{k-1}E^d \,,
\end{align*}
where $E= n^{n_1}b^{n_1-k_1}+(-1)^{n_1+1}(n-k)^{n_1-k_1}k^{k_1}a^{n_1}$ .
\end{theorem}

Trinomials are used to construct irreducible polynomials in \cite{Gathen01, Swan}.
Based on our experiments we give the following conjecture.
\begin{conjecture}\label{conj-trin}
The following polynomial family has at least one polynomial that satisfy property \ref{condition} for $r=2$, 
$$\mathcal{F} \,=\, \{\mathcal{T}_{(2i,k,a,b)} \,\mid\, 1\le i,k,a,b \le \log^2 p \,\} \,.$$
\end{conjecture}
 
We leave the proof, or a refutation, of this conjecture as an open question. 

\subsection{Weaker Generalized Riemann Hypothesis} \label{6.2} 

In 1952, Ankeny \cite{Ankeny} proved that if the Generalized Riemann Hypothesis is true then the least quadratic nonresidue in $\mathbb{F}_p$ is $O(\log^2 p).$
The Generalized Riemann hypothesis(GRH) says that all the non-trivial roots $\rho$ of the Dirichlet L function   are on real line $z =\frac{1}{2}$, but what if  we consider a weaker form  of it? Instead of saying that all the nontrivial roots lie on Re$(\rho)=\frac{1}{2}$, we ``merely'' conjecture that all the nontrivial  roots lie in a wider strip $[\frac{1}{2}- \epsilon, \frac{1}{2} +\epsilon]$, for a constant $\epsilon$. 

\begin{conjecture}[Weak GRH]\label{conj-wgrh}
Let $\chi$ be a Dirichlet character, i.e.~$\chi: \mathbb{F}_p^* \longrightarrow \mathbb{C}^*$.
There exists a constant $\frac{1}{2} > \epsilon \geq 0$ such that the Dirichlet L function $L(s,\chi)=\sum \frac{\chi(n)}{n^s}$ have all its nontrivial roots in the interval $\frac{1}{2}-\epsilon < \text{Re}(s) < \frac{1}{2} + \epsilon$. 
\end{conjecture}

We will now use some known facts from Analytic number theory, for detailed proofs of these facts see \cite[Chap.7]{RM}. Let $\Lambda$ be the Mangoldt function and $\zeta(s)$ be the Riemann zeta function. 

\begin{lemma}[Bounds for $\psi(x,\chi)$]\label{fac1}
Let $\psi(x,\chi)=\sum_{i \leq x}\Lambda(i)\chi(i)$ and $\chi$ be a primitive Dirichlet character of $\mathbb{F}_p^*$ , then 
$$\psi(x,\chi) \,=\, -\sum_{|\gamma| < \sqrt{x}}\frac{x^{\rho}}{\rho} \,+\, O\big(\log^2px \big) \,,$$
where  $\rho=\sigma + i\gamma$ are the nontrivial roots of the Dirichlet  L function $L(s,\chi)$. Also, $\sum_{|\gamma| < \sqrt{x}}\frac{1}{|\rho|} = O(\log^2 px)$.
\end{lemma}
 
\begin{lemma}[Bounds for $\psi(x)$]\label{fac2}
Let $\psi(x)=\sum_{i \leq x}\Lambda(i)$, then
$$ \psi(x) \,=\, x-\sum_{|\gamma| < \sqrt{x}}\frac{x^{\rho}}{\rho} \,+\, O \big(\sqrt{x}\log^2x\big) \,,$$  
where $\rho=\sigma + i\gamma$ are the nontrivial roots of the Riemann zeta function $\zeta(s)$. Also, $\sum_{|\gamma| < \sqrt{x}}\frac{1}{|\rho|} = O(\log^2 x)$.
\end{lemma}  
  
We will now prove bounds on $\psi(x)$ and $\psi(x,\chi)$ assuming Weak GRH.
\begin{lemma}[New bounds]\label{WGRH} 
Assuming Weak GRH, 
\begin{enumerate}
\item $\psi(x,\chi)=O(x^{\frac{1}{2}+\epsilon}\log^2px)$, and
\item $\psi(x)=x+ O(x^{\frac{1}{2}+\epsilon}\log^2x)$.
\end{enumerate} 
\end{lemma}
\begin{proof}
\begin{enumerate}
\item Using the notation in Lem.\ref{fac1},
\begin{align*}
\Big|\sum_{\gamma < \sqrt{x}}\frac{x^{\rho}}{\rho}\Big| &\,\leq\, (\max_{\rho}|x^{\rho}|) \cdot\Big| \sum_{\rho}\frac{1}{\rho}\Big| \\
& \,\leq\,  x^{\frac{1}{2}+\epsilon} \cdot\Big| \sum_{\gamma < \sqrt{x}}\frac{1}{|\rho|}\Big| \\
& \,=\, O(x^{\frac{1}{2}+\epsilon}\log^2px) \,.
\end{align*}
\noindent
Since,  $\psi(x,\chi) = -\sum_{|\gamma| < \sqrt{x}}\frac{x^{\rho}}{\rho}+O\big(\log^2px \big)$, we get that $\psi(x,\chi)=O(x^{\frac{1}{2}+\epsilon}\log^2px)$.
\\
 
\item Using the notation in Lem.\ref{fac2},
\begin{align*}
\Big|\sum_{\gamma < \sqrt{x}}\frac{x^{\rho}}{\rho}\Big| & \,\leq\, (\max_{\rho}|x^{\rho}|) \cdot\Big| \sum_{\rho}\frac{1}{\rho}\Big| \\
& \,\leq\,  x^{\frac{1}{2}+\epsilon} \cdot\Big| \sum_{\gamma < \sqrt{x}}\frac{1}{|\rho|}\Big| \\
& \,=\, O(x^{\frac{1}{2}+\epsilon}\log^2x) \,.
\end{align*}
\noindent
Since, $ \psi(x)=x-\sum_{|\gamma| < \sqrt{x}}\frac{x^{\rho}}{\rho}+ O \big(\sqrt{x}\log^2x\big)$, we get that $\psi(x)=x+ O(x^{\frac{1}{2}+\epsilon}\log^2 x)$. 
\end{enumerate}
\end{proof}

Using this lemma we will bound the least $r$-th nonresidue in $\mathbb{F}_p$. 

\begin{theorem}\label{thm-npr}
Let n(p,r) denote the least $r$-th nonresidue in $\mathbb{F}_p^*$. Then, assuming the Weak GRH,  
$$n(p,r) \,=\, O(\log^{\frac{4}{1-2\epsilon}} p) \,.$$
\end{theorem}
\begin{proof}
Let $\chi_r(a):= a^{\frac{p-1}{r}} \bmod p$, and $\chi_o$ be the trivial character i.e., $\chi_0(a)=1 , \forall a \in \mb{F}_p^*$. Consider 
$$S(M) \,:=\, \sum_{1\leq a\leq M}\chi_o (a)\Lambda(a)-\sum_{1\leq a\leq M}\chi_r(a)\Lambda(a) \,.$$
Note that, S(M) is zero iff there is no $r$-th nonresidue in the initial interval $[M]$.

We have,
\begin{align*}
S(M)& \,=\, \psi(M,\chi_o)-\psi(M,\chi_r) \\
& \,=\, M +O(M^{0.5 + \epsilon}\log^2 pM) && [\mbox{ Using Lem.\ref{WGRH} }]
\end{align*}

We are interested in finding the maximum $M_0$ such that  $S(M_0)=0$. The above estimate implies that $M_0 = O(M_0^{0.5 +\epsilon}\log^2 pM_0)$. 

Therefore,  $n(p,r)= O(\log^{\frac{4}{1-2\epsilon}}p)$.
\end{proof}

This elementary analysis, assuming Weak GRH, has remarkable  consequences. Ankeny's result has been used in derandomizing many computational problems under the assumption of GRH. Some of them are primality testing \cite[Chap.9]{BS}, $r$-th root finding \cite{AMM}, constructing irreducible polynomials over finite fields
\cite{AL86} and cases of polynomial factoring over finite fields \cite{BGL, AL86}. (Also, see \cite{AIKS14, IKS09} and the references therein.) Our result implies that, for derandomizing these problems, proving the Weak GRH suffices. 

\section{Conclusion} \label{sec-concl}

We give a significant generalization of Stickelberger Lemma (Eqn.\ref{SL});  we can construct an $r$-th nonresidue in $\mathbb{F}_q$ given $\zeta_r \in \mathbb{F}_q $ and a polynomial $f$ satisfying Stickelberger property \ref{condition}. Using this, we also gave an algorithm to find $r$-th roots in $\F_{q^m}$ if $r=O(1)$ and $r|\gcd(m,p-1)$. An interesting open question here is whether one can weaken the Stickelberger property (eg.~remove the nondivisibility by $r$ condition?).

Our result along with some known results on finding $\zeta_r \in \mathbb{F}_q$ gives us some interesting applications. It seems that finding $\zeta_r \in  \mathbb{F}_q$ is an inherent requirement in our analysis. We leave removing the requirement of $\zeta_r $ from our algorithm as an open question. This we have been able to achieve, if the goal is only to construct a degree $r$ irreducible (given $f$) instead of an $r$-th nonresidue.

We also leave the concrete conjectures Conj.\ref{conj-trin} \& \ref{conj-wgrh} open.

\section*{Acknowledgements}
Part of research was accomplished while the first two authors were visiting
CQT, NUS.
V.B.~would also like to thank CSE and  IITK for their generous hospitality. N.S.~thanks the funding support from DST (DST/SJF/MSA-01/2013-14). R.M.~would like to thank DST Inspire grant.

\bibliographystyle{abbrv}
\bibliography{qnr} 

\end{document}